\documentclass[10pt,a4paper,oneside]{article}
\usepackage{setspace,graphicx,epstopdf,amsmath,amsfonts,amsgen,amstext,amsthm,amsbsy,amsopn,amssymb,txfonts,bbm,epic, pstricks,tikz,pict2e,parskip}  
\usepackage[square,numbers]{natbib} 
\usepackage[utf8]{inputenc}        

\let\oldsqrt\sqrt
\def\sqrt{\mathpalette\DHLhksqrt}
\def\DHLhksqrt#1#2{%
\setbox0=\hbox{$#1\oldsqrt{#2\,}$}\dimen0=\ht0
\advance\dimen0-0.2\ht0
\setbox2=\hbox{\vrule height\ht0 depth -\dimen0}%
{\box0\lower0.4pt\box2}}

\theoremstyle{definition}

\newtheorem*{duh}{Definition 3.5 (Duhamel correlation function)}

\theoremstyle{remark}

\newtheorem*{Qspin-matrix}{Remark}

\theoremstyle{plain}   
\newtheorem*{MLRO}{Theorem 2.1 (Long-range order)}
\newtheorem*{QNLRO}{Theorem 3.1 (Long-range order for the quantum nematic model)}
\newtheorem*{PLB}{Proposition 3.2}
\newtheorem*{QRP}{Lemma 3.3 (Reflection positivity)}
\newtheorem*{QRP2}{Lemma 3.4 (Reflection positivity for the quantum nematic model)}
\newtheorem*{QIRB}{Lemma 3.6}
\newtheorem*{doublecomm}{Lemma 3.7}

\begin{document}

\title{Long-range order for the spin-1 Heisenberg model with a small antiferromagnetic interaction}
\author{Benjamin Lees\footnote{b.lees@warwick.ac.uk}\\
\small{Department of Mathematics, University of Warwick,
Coventry, CV4 7AL, United Kingdom}
}
\date{}
\maketitle

\begin{abstract}
We look at the general SU(2) invariant spin-1 Heisenberg model. This family includes the well known Heisenberg ferromagnet and antiferromagnet as well as the interesting nematic (biquadratic) and the largely mysterious staggered-nematic interaction. Long range order is proved using the method of reflection positivity and infrared bounds on a purely nematic interaction. This is achieved through the use of a type of matrix representation of the interaction making clear several identities that would not otherwise be noticed. Using the reflection positivity of the antiferromagnetic interaction one can then show that the result is maintained if we also include an antiferromagnetic interaction that is sufficiently small.
\end{abstract}

\section{Introduction}

Showing the existence of phase transitions at low temperatures for Heisenberg models is a well known difficult problem. There have been several positive results in this area over the years in both the classical and quantum cases. The first rigorous proof of a phase transition in a Heisenberg model was the result of Fr\"ohlich, Simon and Spencer \cite{F-S-S} for the classical Heisenberg ferromagnet (and hence for the antiferromagnet also as it is equivalent to the ferromagnet in the classical case). The result was later extended to the quantum antiferromagnet by Dyson, Lieb and Simon \cite{D-L-S}. The case of spin-$1/2$ in dimension three was not covered, the result was extended to this case by Kennedy, Lieb and Shastry \cite{K-L-S}. The result also shows long-range order for dimension two at zero temperature. In the nematic case (also called the biquadratic interaction) there is known to be a phase transition. In the classical system there is nematic order (also called quadrupolar long-range order) at low 
temperatures, as was shown by Angelescu and Zagrebnov \cite{A-Z}. By contrast for the quantum case there 
is known to be N\'eel order, as was 
recently proved in the work of Ueltschi \cite{D2}. The paper extended and combined the works of T\'oth \cite{T} and Aizenmann and Nachtergaele \cite{A-N} who introduced probabilistic representations of some quantum Heisenberg models. This work also showed the existence of nematic order in a region with an extra ferromagnetic interaction.  All of these results apply in dimension at least three for positive temperature. In dimensions one and two there is the famous result of Mermin and Wagner \cite{M-W} that rules out a phase transition at positive temperature, this does not contradict the result for dimension two in \cite{K-L-S}. For the ground state there are some rigorous results, the work of Tanaka, Tanaka and Idogaki shows long range order for an antiferromagnetic interaction accompanied by a small enough nematic (biquadratic) interaction in dimensions two and three. In dimension three they also show long-range order in part of the nematic region investigated in \cite{D2}, these results were obtained 
independently.
The aim of this article is to show that there is also a phase transition in a region with a nematic interaction accompanied by a small antiferromagnetic interaction, this result was already expected, although an explicit proof has not been presented before. Curiously the result only shows the existence of nematic order, weaker than the expected antiferromagnetic order, this implies that there is further work to be done to strengthen the result to the full antiferromagnetic order.
\\
The positive results concerning long-range order above use the method of reflection positivity in order to obtain an infrared bound, that is, a bound on the Fourier transform of the correlation in question. One can then easily show that the correlation function does not decay (for example that $\langle S_x^3S_y^3\rangle\geq c>0$ uniformly) if the infrared bound is sufficiently strong. The infrared bound proven in \cite{D-L-S} allows to show a phase transition for the antiferromagnet. It is straightforward to extend this result to a model with an antiferromagnetic interaction accompanied by a small nematic (biquadratic) interaction. However when the nematic interaction is too large the result will no longer apply. This article will follow the approach of \cite{D-L-S}, starting with the nematic model, obtaining a lower bound that involves some other correlation functions. This bound can be shown to be positive for low temperatures by relating these correlations to probabilities in the random loop model 
introduced in \cite{A-N}. It is then easy to show (due to reflection positivity of the antiferromagnet interaction) that adding an antiferromagnetic interaction will maintain the positivity of the lower bound, providing the interaction is small enough.

\section{The Spin-1 SU(2)-invariant model}

Let $S\in\frac{1}{2}\mathbb{N}$. For a spin-$S$ model we have local Hilbert spaces $\mathcal{H}_x=\mathbb{C}^{2S+1}$. Observables are then Hermitian matrices built from linear combinations of tensor products of operators on $\otimes_{x\in\Lambda}\mathcal{H}_x$ for some set of sites $\Lambda$. Physically important observables can often be expressed in terms of \emph{spin matrices} $S^1,S^2$ and $S^3$, operators on $\mathbb{C}^{2S+1}$ that are the generators of a (2$S$+1)-dimensional irreducible unitary representation of SU(2) such that
\begin{equation}\label{spins}
 \left[S^\alpha,S^\beta\right]=i\sum_\gamma\mathcal{E}_{\alpha\beta\gamma}S^\gamma
\end{equation}
where $\alpha,\beta,\gamma\in\{1,2,3\}$ and $\mathcal{E}_{\alpha\beta\gamma}$ is the Levi-Civita symbol. Denote $\mathbf{S}=(S^1,S^2,S^3)$, its magnitude is then $\mathbf{S}\cdot\mathbf{S}=S(S+1)\mathbbm{1}$. The case $S=\frac{1}{2}$ gives the Pauli spin matrices. For $S=1$ there are several choices for spin matrices, to make things concrete we will use the following matrices for $S=1$:
\begin{equation}
S^1=\frac{1}{\sqrt{2}}\left(\begin{matrix} 0 & 1 & 0 \\  1 & 0 & 1 \\ 0 & 1 & 0  \end{matrix}\right),\quad S^2=\frac{1}{\sqrt{2}}\left(\begin{matrix} 0 & -i & 0 \\  i & 0 & -i \\  0 & i & 0  \end{matrix}\right),\quad S^3=\left(\begin{matrix} 1 & 0 & 0 \\ 0 & 0 & 0 \\ 0 & 0 & -1 \end{matrix}\right). 
\end{equation}

\par
Consider a pair $(\Lambda,\mathcal{E})$ of a lattice $\Lambda\subset\mathbb{Z}^d$ and a set of edges $\mathcal{E}$ between points in $\Lambda$. Here we will take
\begin{equation}
\Lambda =\left\{-\frac{L}{2}+1,...,\frac{L}{2}\right\}^d,
\end{equation}
for integer $L$. For the set of edges $\mathcal{E}$ we take nearest-neighbour with periodic boundary conditions. Then we take the operator $S^i_x$ for $i=1,2,3$ to be shorthand for the operator $S^i_x\otimes Id_{\Lambda\setminus\{x\}}$.

The Hamiltonian of interest is general the Spin-1 SU(2)-invariant Hamiltonian with a two-body interaction, it is known that this can be written as
\begin{equation}
H^{J_1,J_2}_{\Lambda,\mathbf{0}}=-2\sum_{\{x,y\}\in\mathcal{E}}\left(J_1\left(\mathbf{S}_x\cdot\mathbf{S}_y\right)+J_2\left(\mathbf{S}_x\cdot\mathbf{S}_y\right)^2\right).
\end{equation}
The phase diagram for this model is only partially understood. If $J_2=0$ and $J_1<0$ we have the Heisenberg antiferromagnet that is known to undergo a phase transition at low temperatures \cite{D-L-S}. As the interaction when $J_2>0$ is reflection positive it is also possible to extend this result to $J_2>0$ when the ratio $J_1/J_2$ is sufficiently small. The line $J_1=0$ has been shown to exhibit N\'eel order for low temperatures when $J_2>0$ \cite{D2}, for $J_2<0$ there are no rigorous results, it would be a challenging task to obtain results. The line $J_2=J_1/3<0$ is the AKLT model \cite{A-K-L-T}.

\begin{centering}
\setlength{\unitlength}{6mm}
\begin{figure}[t!]

\includegraphics[width=12cm]{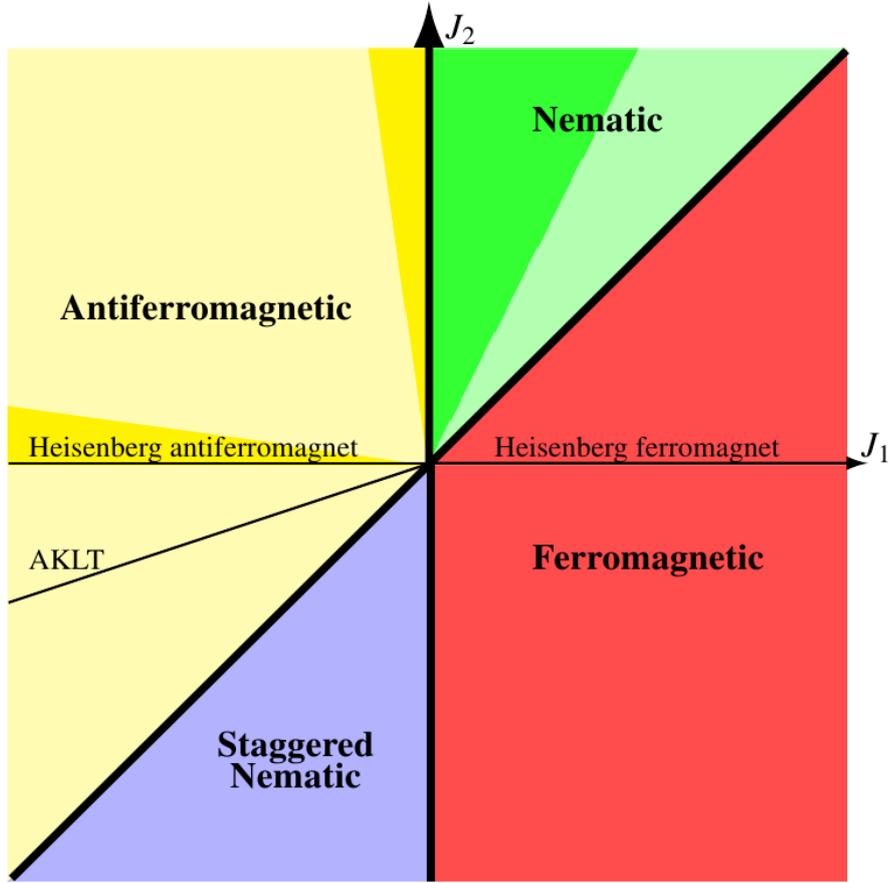}

\caption{\footnotesize{The phase diagram for the general SU(2) invariant spin-1 model. Some regions have rigorous proofs that the expected order is indeed correct. The line $J_1<0$, $J_2=0$ is the Heisenberg antiferromagnet where antiferromagnetic order has been proven \cite{D-L-S}, this region extends slightly into the dark yellow region. The dark green region has nematic order at low temperatures \cite{D2}, with N\'eel order on the line $J_2>0$, $J_1=0$, the adjacent dark yellow region also has long range order, however only the nematic correlation function has been shown not to decay, antiferromagnet order is expected here but is not yet proved.}}
\end{figure}
\end{centering}

\par
The main result of this paper is to show that there is a phase transition in this model for $J_2>0$ and $J_1<0$ with $|J_1|$ sufficiently small compared to $|J_2|$, the statement will be made precise below.
\par
First we define the partition function and Gibbs states of our model as
\begin{align}
Z^{J_1,J_2}_{\beta,\Lambda,\mathbf{0}}=&Tr e^{-\beta H^{J_1,J_2}_{\Lambda,\mathbf{0}}},
\\
\langle\cdot\rangle^{J_1,J_2}_{\beta,\Lambda,\mathbf{0}}=&\frac{1}{Z^{J_1,J_2}_{\beta,\Lambda,\mathbf{0}}}Tr \cdot e^{-\beta H^{J_1,J_2}_{\Lambda,\mathbf{0}}}.
\end{align}
Where $\beta>0$ is the inverse temperature. The quantity of interest is then the correlation
\begin{equation}
\rho(x)=\left\langle\left((S_0^3)^2-\frac{2}{3}\right)\left((S_x^3)^2-\frac{2}{3}\right)\right\rangle^{J_1,J_2}_{\beta,\Lambda,\mathbf{0}}.
\end{equation}
this correlation is specifically of interest for spin-1, in general spin-S $\frac{2}{3}$ will be replaced with $\frac{1}{3}S(S+1)$. The result is then given by the following theorem.
\\
\begin{MLRO}
Let $S=1$, $J_2>0$ and $L$ be even, $d\geq 6$. Then there exists $J_1^0<0$, $\beta_0$ and $C=C(\beta,J_1)>0$ such that if $J_1^0<J_1\leq 0$ and $\beta>\beta_0$ then
\[
\frac{1}{|\Lambda|}\sum_{x\in\Lambda}\rho(x)\geq C.
\]
for all $L$ large enough.
\end{MLRO}
The proof of the result will be in two steps, first the result will be proved for $J_1=0$, this will be the content of the next section. Second it will be shown how the result for $J_1=0$ extends to sufficiently small $J_1<0$, this should come as no surprise as the interaction is reflection positive for $J_1<0$ hence adding a small interaction in this direction should not alter the result too much.

\section{The model $J_2>0$, $J_1=0$}

We will now consider the so-called quantum nematic model $J_2>0$, $J_1=0$, the aim is to prove long-range order for this model using a similar approach to the proofs in \cite{D-L-S,F-I-L-S,F-I-L-S2,F-S-S}. To do this we will use a representation that is an analogue of the matrix representation used in \cite{A-Z}. Care must be taken as now we are working with matrices rather than vectors and so commutativity becomes an issue. We introduce an external field, $\mathbf{h}$, to the Hamiltonian
\begin{equation}\label{QNHamiltonian}
 H^{0,1}_{\Lambda,\mathbf{h}}=-2\sum_{\{x,y\}\in\mathcal{E}}(\mathbf{S}_x\cdot\mathbf{S}_y)^2-\sum_{x\in\Lambda}h_x\left((S_x^3)^2-\frac{1}{3}S(S+1)\mathbbm{1}\right).
\end{equation}
Here $\mathbbm{1}$ is the identity matrix. Equilibrium states are given by
\begin{equation}
\langle A\rangle^{0,1}_{\beta,\Lambda,\mathbf{h}}=\frac{1}{Z^{0,1}_{\beta,\Lambda,\mathbf{h}}}Tr Ae^{-\beta H^{0,1}_{\Lambda,\mathbf{h}}}.
\end{equation}
Note that the $J_2$ has been absorbed into the parameter $\beta$. Using the direct analogue of \cite{A-Z} will not work here, the reason is that reflection positivity will fail as $\overline{S^2}=-S^2$. All other attempts to directly obtain a matrix representation of the interaction $(\mathbf{S}_x\cdot\mathbf{S}_y)^2$ have also failed, however, there is a solution. We will instead use a matrix representation of a Hamiltonian that is unitarily equivalent to \eqref{QNHamiltonian}.
\par
From now on we will work with the following Hamiltonian
\begin{equation}\label{QNHamiltonian2}
  H^U_{\Lambda,\mathbf{h}}=-2\sum_{\{x,y\}\in\mathcal{E}}(S_x^1S_y^1-S_x^2S_y^2+S_x^3S_y^3)^2-\sum_{x\in\Lambda}h_x\left((S_x^3)^2-\frac{1}{3}S(S+1)\mathbbm{1}\right),
\end{equation}
and partition function
\begin{equation}\label{QNpartition}
 Z^U_{\Lambda,\beta,\mathbf{h}}=Tr e^{-\beta H^U_{\Lambda,\mathbf{h}}}.
\end{equation}
Similarly to before, equilibrium states are given by
\begin{equation}
\left\langle A\right\rangle^U_{\Lambda,\beta,\mathbf{h}}=\frac{1}{Z^U_{\Lambda,\beta,\mathbf{h}}}Tr Ae^{-\beta H^U_{\Lambda,\mathbf{h}}}.
\end{equation}

If $\Lambda$ has a bipartite structure, $\Lambda=\Lambda_A\cup\Lambda_B$, then if we define $U=\prod_{x\in\Lambda_B}e^{i\pi S_x^2}$ we have 
\begin{equation}
 U^{-1}H^U_{\Lambda,\mathbf{h}}U=H^{0,1}_{\Lambda,\mathbf{h}}.
\end{equation}
Note that this leaves $\rho(x)$ unchanged. Before the theorem we introduce an integral, it is also introduced in \cite{K-L-S},
\begin{equation}
I_d=\frac{1}{(2\pi)^d}\int_{[-\pi,\pi]^d}\sqrt{\frac{\varepsilon(k+\pi)}{\varepsilon(k)}}\left(\frac{1}{d}\sum_{i=1}^d\cos k_i\right)_+\mathrm{d}k,
\end{equation}
where
\begin{equation}
\varepsilon(k)=2\sum_{i=1}^d\left(1-\cos k_i\right).
\end{equation}
We have $I_d<\infty$ for $d\geq 3$ and it can be shown that $I_d\to 0$ as $d\to\infty$ \cite{K_L_S}. Then we have the following result:
\newline
\begin{QNLRO}
 Let $S=1$. Assume $\mathbf{h}=0$ and $L$ is even with $d\geq 3$. Then we have the bound
 \[ 
  \lim_{\beta\to\infty}\lim_{L\to\infty}\frac{1}{|\Lambda|}\sum_{x\in\Lambda}\rho(x) \geq \rho(e_1)-I_d\sqrt{\left\langle S_0^1S_0^3S_{e_1}^1S_{e_1}^3\right\rangle^U_{\mathbb{Z}^d,\infty,\mathbf{h}}}.
 \]
\end{QNLRO}

The expectations on the right of the inequality are taken in the infinite volume limit and with $\beta\to\infty$. If this lower bound is strictly positive it implies the existence of a phase transition at low temperatures, note that the lower bound is valid in any dimension $d\geq 3$, but as can be seen from equation \eqref{disclb} not in $d\leq 2$, hence no phase transition. This is consistent with the well known Mermin-Wagner theorem \cite{M-W}. Using the loop model introduced in \cite{A-N} and extended in \cite{D2} we can relate the expectations in the lower bound to the probability of the event $E_{0,e_1}$, that two nearest neighbours are in the same loop as
\begin{equation}\label{prob}
\rho(e_1)=\frac{2}{9}\mathbb{P}\left[E_{0,e_1}\right],\qquad \left\langle S_0^1S_0^3S_{e_1}^1S_{e_1}^3\right\rangle^U_{\Lambda,\beta,\mathbf{h}}=\frac{1}{3}\mathbb{P}\left[E_{0,e_1}\right].
\end{equation}
So we can write the lower bound as $\sqrt{\mathbb{P}\left[E_{0,e_1}\right]}\left(\frac{2}{9}\sqrt{\mathbb{P}\left[E_{0,e_1}\right]}-\frac{I_d}{\sqrt{3}}\right)$. This means a sufficiently large lower bound on $\mathbb{P}\left[E_{0,e_1}\right]$ will allow to show the lower bound is positive in high enough dimension.
\\
\begin{PLB}
For $d\geq 1$, $S=1$ and $L_1=...=L_d=L$ even. We have the lower bound
\begin{equation}
\mathbb{P}\left[E_{0,e_1}\right]\geq\frac{1}{4}.
\end{equation}
\end{PLB}
Putting this bound into the theorem and computing $I_d$ for various $d$ shows that there is a positive lower bound (and hence phase transition) for $d\geq 6$.
\begin{proof}
For any state $\psi\in\otimes_{x\in\Lambda}\mathbb{C}^3$ we have that in the ground state
\begin{equation}\label{trialineq}
\langle H^{0,1}_{\Lambda,\mathbf{0}}\rangle^{0,1}_{\Lambda,\infty,\mathbf{h}}\leq\langle \psi,H^{0,1}_{\Lambda,\mathbf{0}}\psi\rangle.
\end{equation}
We pick the N\'eel state, $\psi_{N\acute{e}el}$ as a trial state
\begin{equation}
\psi_{N\acute{e}el}=\otimes_{x\in\Lambda}|(-1)^x\rangle.
\end{equation}
We have used Dirac notation here where $S^3|a\rangle=a|a\rangle$. For the left of \eqref{trialineq} we recall that for $x$ and $y$ nearest neighbours $(\mathbf{S}_x\cdot\mathbf{S}_y)^2$ has three terms of the form $(S_x^i)^2(S_y^i)^2$, having expectation $\frac{2}{9}\mathbb{P}\left[E_{0,e_1}\right]+\frac{4}{9}$ independent of $i$ and six terms of the form $S_x^iS_x^jS_y^iS_y^j$ having expectation $\frac{1}{3}\mathbb{P}\left[E_{0,e_1}\right]$ independent of $i$ and $j$ (this is due to the equivalent roles of $i$ and $j$ coupled with $(S_x^iS_x^j)^T=\pm S_x^jS_x^i$ where the sign depends on the value of $i$ or $j$). This gives 
\begin{equation}
\langle H^{0,1}_{\Lambda,\mathbf{0}}\rangle^{0,1}_{\mathbb{Z}^d,\infty,
\mathbf{h}}=-2\sum_{\{x,y\}\in\Lambda}\left[3\left(\frac{2}{9}\mathbb{P}\left[E_{0,e_1}\right]+\frac{4}{9}\right)+\frac{6}{3}\mathbb{P}\left[E_{0,e_1}\right]\right]=-8d|\Lambda|\frac{2\mathbb{P}\left[E_{0,e_1}\right]+1}{3}.
\end{equation}
For the right side of \eqref{trialineq} it can be checked that, for $S=1$, $(\mathbf{S}_x\cdot\mathbf{S}_y)^2=P_{xy}+1$ where $\frac{1}{3}P_{xy}$ is the projector onto the spin singlet. Hence
\begin{equation}
\langle 1,-1 |(\mathbf{S}_x\cdot\mathbf{S}_y)^2|1,-1\rangle=\langle 1,-1|P_{x,y}+1|1,-1\rangle=2,
\end{equation}
from this we see that the right side of \eqref{trialineq} is $-4d|\Lambda|$. Inserting each of these values into \eqref{trialineq} and rearranging gives the claim of the proposition.
\end{proof}
Note that if one could find a state with lower energy than the N\'eel state this lower bound could be improved and hence potentially the theorem strengthened to show phase transitions in lower dimensions. However the problem of finding lower energy states does not appear an easy one.

 The rest of the section will be dedicated to the proof of theorem 3.1. We will proceed with calculations for general spin until it becomes necessary to restrict to the case $S=1$. Fortunately for this Hamiltonian we can find a matrix representation. Define $Q_x$ as
\begin{equation}\label{Qspin-matrix}
 Q_x=\left(\begin{matrix}
  (S_x^1)^2-\frac{1}{3}S(S+1) & S_x^1 iS_x^2 & S_x^1 S_x^3 \\
  S_x^1 iS_x^2 & (S_x^2)^2-\frac{1}{3}S(S+1) & iS_x^2 S_x^3 \\
  S_x^1 S_x^3 & iS_x^2 S_x^3 & (S_x^3)^2-\frac{1}{3}S(S+1)
 \end{matrix}\right).
\end{equation}
We introduce the operation $\mathcal{TR}$, which is the sum of diagonal entries of matrices of the form of $Q_x$, however this `trace' will return an operator, not a number, so we distinguish it from the normal trace. As an example we see that $\mathcal{TR}(Q_x)=0$, the zero matrix.  We have the relation (note that below we do \emph{not} mean `normal' matrix multiplication, we only write $Q_xQ_y$ for convenience as explained in the remark).
\begin{equation}\label{matrixid}
 \mathcal{TR} (Q_xQ_y)=(S_x^1S_y^1-S_x^2S_y^2+S_x^3S_y^3)^2-\frac{1}{3}S^2(S+1)^2\mathbbm{1}.
\end{equation}
\begin{Qspin-matrix}
 We must be careful here, as we are working with a matrix of matrices, as to what we mean by multiplication. The representation \eqref{Qspin-matrix} is not at all essential to the proof, the advantage of using it is that once \eqref{matrixid} has been verified other relations can be stated much more concisely and clearly and easily checked, these relations are not at all obvious or easy to come up with without using \eqref{matrixid}. 
 
 By the product $Q_xQ_y$ we follow the `normal' matrix multiplication with the added stipulation that for the $i$\textsuperscript{th} diagonal entry of $Q_xQ_y$ the operator $S^i$ will appear first. For example in entry $\{1,1\}$ of $Q_xQ_y$ there is the term $S_x^1iS_x^2S_y^1iS_y^2$, in the entry $\{2,2\}$ this term will become $iS_x^2S_x^1iS_y^2S_y^1$, this ensures that we have each of the cross terms in the right-hand side of \eqref{matrixid}. For off-diagonal entries we are not concerned as we are always taking a `trace'.

 In the case $x\neq y$ less care is needed as components of $\mathbf{S_x}$ and $\mathbf{S_y}$ commute (in fact $\mathcal{TR} Q_xQ_y=\mathcal{TR} Q_yQ_x$, hence we must only take care that the product order of components of spin at the same site is maintained).
\end{Qspin-matrix}
We also have that $\mathcal{TR} Q_x^2=C^S_{x}-\frac{1}{3}S^2(S+1)^2$ acting on $\mathcal{H}_x$. In $S= 1$
\begin{equation}
 C^1_x=\left(\begin{matrix}
  2 & 0 & 2 \\
  0 & 0 & 0 \\
  2 & 0 & 2
 \end{matrix}\right)_x.
\end{equation}
Using this we can represent our interaction as
\begin{equation}\label{QNRelation}
 (S_x^1S_y^1-S_x^2S_y^2+S_x^3S_y^3)^2=\frac{1}{2}\left(C^S_x+C^S_y-\mathcal{TR}\left[(Q_x-Q_y)^2\right]\right).
\end{equation}
We introduce the field $v$ with value $v_x\in\mathbb{R}$ at the site $x\in\Lambda$. We denote by $\mathbf{v}$ the field of $3\times3$ matrices such that each $\mathbf{v}_x$ has one non-zero entry, the entry $\{3,3\}$ being $v_x\in\mathbb{R}$. We define
\begin{align}
 &H(v)=\sum_{\{x,y\}\in\mathcal{E}}\left(\mathcal{TR}\left[(Q_x-Q_y)^2\right]-C^S_x-C^S_y\right)-\sum_{x\in\Lambda}(\Delta v)_x\left((S_x^3)^2-\frac{1}{3}S(S+1)\right),
 \\
 &Z(v)=Tr e^{-\beta H(v)}.
\end{align}
Note that from \eqref{QNRelation} $H(v)=H^U_{\Lambda,\Delta v}$. Here we have used the lattice Laplacian and below we use the inner product $(f,g)=\sum_{x\in\Lambda}f_xg_x$ with the identity $(f,-\Delta g)=\sum_{\{x,y\}\in\mathcal{E}}(f_x-f_y)(g_x-g_y)$.
Then we can calculate as follows:
\begin{equation}
 \begin{aligned}
  H(v)=\sum_{\{x,y\}\in\mathcal{E}}\Bigg\{\mathcal{TR}&\left[(Q_x+\frac{\mathbf{v}_x}{2}-Q_y-\frac{\mathbf{v}_y}{2})^2\right]-\mathcal{TR}\left[(Q_x-Q_y)(\mathbf{v}_x-\mathbf{v}_y)\right]
  \\
  &-C^S_x-C^S_y+(v_x-v_y)\left((S_x^3)^2-(S_y^3)^2\right)-\frac{1}{4}(v_x-v_y)^2\Bigg\}
  \\
  =\sum_{\{x,y\}\in\mathcal{E}}\Bigg\{\mathcal{TR}&\left[(Q_x+\frac{\mathbf{v}_x}{2}-Q_y-\frac{\mathbf{v}_y}{2})^2\right]-C^S_x-C^S_y\bigg\}-\frac{1}{4}(v,-\Delta v).
 \end{aligned}
\end{equation}
We must check carefully when dealing with the cross terms $(Q_x-Q_y)(\mathbf{v}_x-\mathbf{v}_y)$ and $(\mathbf{v}_x-\mathbf{v}_y)(Q_x-Q_y)$, they are not equal but $\mathcal{TR}(Q_x-Q_y)(\mathbf{v}_x-\mathbf{v}_y)= \mathcal{TR}(\mathbf{v}_x-\mathbf{v}_y)(Q_x-Q_y)$, so the calculation is correct. From this it makes sense to define the following Hamiltonian and partition function:
\begin{align}
 &H'(v)=H(v) +\frac{1}{4}(v,-\Delta v),
 \\
 &Z'(v)=Tr e^{-\beta H'(v)}.
\end{align}
Now the property of Guassian Domination is
\begin{equation}\label{GD}
 Z(v)\leq Z(0)e^{\frac{\beta}{4}(v,-\Delta v)} \iff Z'(v)\leq Z'(0),
\end{equation}
as in the classical case it follows from reflection positivity.
\newline
\begin{QRP}
 Let $\mathcal{H}=\mathit{h}\otimes\mathit{h}$, $dim\,\mathit{h}<\infty$, fix a basis. Let $A,B,C_i,D_i$ for $i=1,...,k$ be matrices in $\mathit{h}$, then
 \begin{equation}
 \begin{aligned}
  \Bigg|Tr_{\mathcal{H}}\exp\bigg\{A\otimes \mathbbm{1}+&\mathbbm{1}\otimes B-\sum_{i=1}^k(C_i\otimes \mathbbm{1}-\mathbbm{1}\otimes D_i)^2\bigg\}\Bigg|^2
  \\
  \leq Tr&_{\mathcal{H}}\exp\Bigg\{A\otimes \mathbbm{1}+\mathbbm{1}\otimes\bar{A}-\sum_{i=1}^k(C_i\otimes \mathbbm{1}-\mathbbm{1}\otimes\bar{C_i})^2\Bigg\}
  \\
  \times
  &Tr_{\mathcal{H}}\exp\Bigg\{\bar{B}\otimes \mathbbm{1}+\mathbbm{1}\otimes B-\sum_{i=1}^k(\bar{D_i}\otimes \mathbbm{1}-\mathbbm{1}\otimes D_i)^2\Bigg\}  
  \end{aligned}
 \end{equation}
where $\bar{A}$ is the complex conjugate of $A$.
\end{QRP}
The proof uses Trotter's formula. As in the classical case, reflection positivity is a very powerful tool, for more information see \cite{B-C,D-L-S,F-I-L-S,F-I-L-S2,F-S-S,T-P,U-T,D2}. \par
Before we prove reflection positivity for our partition function we should calculate the trace in $Z'(v)$, recall how we have defined our multiplication.
\begin{equation}\label{Qtrace}
\begin{aligned}
 \mathcal{TR}\left[(Q_x+\frac{\mathbf{v}_x}{2}-Q_y-\frac{\mathbf{v}_y}{2})^2\right]=&\left((S_x^1)^2-(S_y^1)^2\right)^2+\left((S_x^2)^2-(S_y^2)^2\right)^2
 \\
 &+\left((S_x^3)^2+\frac{\mathbf{v}_x}{2}-(S_y^3)^2-\frac{\mathbf{v}_y}{2}\right)^2+\left(S_x^1iS_x^2-S_y^1iS_y^2\right)^2
 \\
&+\left(S_x^1S_x^3-S_y^1S_y^3\right)^2+\left(iS_x^2S_x^3-iS_y^2S_y^3\right)^2
\\ &+\left(iS_x^2S_x^1-iS_y^2S_y^1\right)^2+\left(S_x^3S_x^1-S_y^3S_y^1\right)^2 +\left(S_x^3iS_x^2-S_y^3iS_y^2\right)^2.
 \end{aligned}
\end{equation}
Now we have enough information to use the Lemma, let $R:\Lambda\to\Lambda$ be a reflection that swaps $\Lambda_1$ and $\Lambda_2$ where $\Lambda=\Lambda_1\cup\Lambda_2$, each such reflection defines two sub-lattices of $\Lambda$ in this way, we split the field $v=(v_1,v_2)$ on the sub-lattices $\Lambda_1$ and $\Lambda_2$.
\newline
\begin{QRP2}
 For $S\in\frac{1}{2}\mathbb{N}$ and any reflection, $R$, across edges and $v=(v_1,v_2)$
 \[
  Z((v_1,v_2))^2\leq Z((v_1,Rv_1))Z((Rv_2,v_2)).
 \]
\end{QRP2}
\begin{proof}
We cast $Z'(v)$ in RP form. Let
\begin{equation}
 \begin{aligned}
  A=&-\beta\sum_{\{x,y\}\in\mathcal{E}_1}\mathcal{TR}\left[(Q_x+\frac{\mathbf{v}_x}{2}-Q_y-\frac{\mathbf{v}_y}{2})^2\right]-\beta d\sum_{x\in\Lambda_1}C^S_x,
  \\
  B=&\text{same in }\Lambda_2,
 \end{aligned}
\end{equation}
where $\mathcal{E}_1$ is the set of edges in $\Lambda_1$ and we note that the term $C^S_x$ occurs $d$ times in the sum over $\mathcal{E}$ for each $x\in\Lambda$. Further define
\begin{equation}
 \begin{array}{l l}
  C_i^1=\sqrt{\beta}(S_{x_i}^1)^2,\quad \quad                         & D_i^1=\sqrt{\beta}(S_{y_i}^1)^2. \\
  C_i^2=\sqrt{\beta}(S_{x_i}^2)^2,  \quad \quad                       & D_i^2=\sqrt{\beta}(S_{y_i}^2)^2. \\
  C_i^3=\sqrt{\beta}((S_{x_i}^3)^2+\frac{\mathbf{v}_{x_i}}{2}),  \quad \quad    & D_i^3=\sqrt{\beta}((S_{y_i}^3)^2+\frac{\mathbf{v}_{y_i}}{2}). \\
  C_i^4=\sqrt{\beta}S_{x_i}^1iS_{x_i}^2 ,  \quad \quad                & D_i^4=\sqrt{\beta}S_{y_i}^1iS_{y_i}^2.\\
  C_i^5=\sqrt{\beta}S_{x_i}^1S_{x_i}^3, \quad   \quad                 & D_i^5=\sqrt{\beta}S_{y_i}^1S_{y_i}^3. \\
  C_i^6=\sqrt{\beta}iS_{x_i}^2S_{x_i}^3,  \quad \quad                 & D_i^6=\sqrt{\beta}iS_{y_i}^2S_{y_i}^3. \\
  C_i^7=\sqrt{\beta}iS_{x_i}^2S_{x_i}^1 ,  \quad \quad                & D_i^7=\sqrt{\beta}iS_{y_i}^2S_{y_i}^1.\\
  C_i^8=\sqrt{\beta}S_{x_i}^3S_{x_i}^1, \quad   \quad                 & D_i^8=\sqrt{\beta}S_{y_i}^3S_{y_i}^1. \\
  C_i^9=\sqrt{\beta}S_{x_i}^3iS_{x_i}^2,  \quad \quad                 & D_i^9=\sqrt{\beta}S_{y_i}^3iS_{y_i}^2.
 \end{array}
\end{equation}
Where $\{x_i,y_i\}$ are edges crossing the reflection plane with $x_i\in\Lambda_1$ and $y_i\in\Lambda_2$. Because $\overline{S_x^1}=S_x^1$, $\overline{S_x^3}=S_x^3$, $\overline{iS_x^2}=iS_x^2$ we see from the previous lemma that $Z'((v_1,v_2))^2\leq Z'((v_1,Rv_1))Z'((Rv_2,v_2))$, from which the result follows.
\end{proof}
The Gaussian domination inequality \eqref{GD} follows from this just as in the classical case, a proof can be found in \cite{D-L-S}. The next step in the classical case was to obtain an infrared bound for the correlation function $\rho(x)$, we cannot do this directly but we can obtain an infrared bound for the \emph{Duhamel correlation function}.
\newline
\begin{duh}
 For matrices $A,B$ we define the \emph{Duhamel correlation function} $(A,B)_{Duh}$ as 
 \[
  (A,B)_{Duh}=\frac{1}{Z(0)}\frac{1}{\beta}\int_0^\beta\mathrm{d}sTrA^*e^{-sH(0)}Be^{-(\beta-s)H(0)}
 \]
Note that this is an inner product.
\end{duh}
Now to use this correlation function we must first fix our definition of the Fourier transform
\begin{equation}
\begin{aligned}
 \mathcal{F}(f)(k)=\hat{f}(k)=&\sum_{x\in\Lambda}e^{-ikx}f(x) \quad\quad\,\,\,\,\,\,\, k\in\Lambda^*,
 \\
 f(x)=&\frac{1}{|\Lambda|}\sum_{k\in\Lambda^*}e^{ikx}\hat{f}(k) \quad\quad x\in\Lambda.
 \end{aligned}
\end{equation}
where
\begin{equation}
 \Lambda^*=\frac{2\pi}{L_1}\left\{-\frac{L_1}{2}+1,...,\frac{L_1}{2}\right\}\times...\times\frac{2\pi}{L_d}\left\{-\frac{L_d}{2}+1,..,\frac{L_d}{2}\right\},
\end{equation}
\begin{QIRB}
 For $S\in\frac{1}{2}\mathbb{N}$ and $L_i$ even for $i=1,...,d$ we have the following infrared bound
 \begin{equation}
  \mathcal{F}\left((S_0^3)^2-\frac{1}{3}S(S+1),(S_x^3)^2-\frac{1}{3}S(S+1)\right)_{Duh}(k)\leq \frac{1}{2\beta\varepsilon(k)}.
 \end{equation}
\end{QIRB}
\begin{proof}
 We begin as usual by choosing $v_x=\eta\cos(kx)$, then from Taylor's theorem and using $h=\Delta\mathbf{v}=-\varepsilon(k)\mathbf{v}$ we see
 \begin{equation}
  Z(\mathbf{v})=Z(0)+\frac{1}{2}\left(h,\left.\frac{\partial^2Z(\mathbf{v})}{\partial h_x\partial h_y}\right|_{h=0}h\right)+O(\eta^4).
 \end{equation}
Using the Duhamel formula 
\begin{equation}
 e^{\beta(A+B)}=e^{\beta A}+\int_0^\beta\mathrm{d}s e^{sA}Be^{(\beta-s)(A+B)}
\end{equation}
 with $A=H(0)$ and $B=-\sum_{x\in\Lambda}(\Delta v)_x\left((S_x^3)^2-\frac{1}{3}S(S+1)\right)$ gives
\begin{equation}
 \frac{1}{Z(0)}\left.\frac{\partial^2Z(\mathbf{v})}{\partial h_x\partial h_y}\right|=\beta^2\left((S_x^3)^2-\frac{1}{3}S(S+1),(S_y^3)^2-\frac{1}{3}S(S+1)\right)_{Duh}.
\end{equation}
Putting this together we have
\begin{equation}
 \begin{aligned}
  Z(&\mathbf(v))-O(\eta^4)=
  \\
  &Z(0)+\frac{1}{2}Z(0)(\eta\varepsilon(k)\beta)^2\sum_{x,y\in\Lambda}\cos(kx)\cos(ky)\left((S_x^3)^2-\frac{1}{3}S(S+1),(S_y^3)^2-\frac{1}{3}S(S+1)\right)_{Duh}
  \\
  =&Z(0)+\frac{1}{2}Z(0)\beta^2\eta^2\varepsilon(k)^2\mathcal{F}\left((S_0^3)^2-\frac{1}{3}S(S+1),(S_y^3)^2-\frac{1}{3}S(S+1)\right)_{Duh}\sum_{x\in\Lambda}\cos^2(kx).
 \end{aligned}
\end{equation}
Also
\begin{equation}
 e^{-\frac{1}{4}\beta(\mathbf{v},\Delta\mathbf{v})}=e^{\frac{1}{4}\beta\varepsilon(k)\eta^2\sum\cos^2(kx)},
\end{equation}
comparing the order $\eta^2$ terms gives the result.
\end{proof}
To transfer the infrared bound to the normal correlation function we would like to use the Falk-Bruch inequality \cite{F-B}:
\begin{equation}
 \frac{1}{2}\langle A^*A+AA^*\rangle \leq(A,A)_{Duh}+\frac{1}{2}\sqrt{(A,A)_{Duh}\langle[A^*,[H^U_{\Lambda,\mathbf{h}},A]]\rangle}
\end{equation}
where is the Hamiltonian of the system. If we attempt to use this inequality with $A=\mathcal{F}\left((S_x^3)^2-\frac{1}{3}S(S+1)\right)(k)$ and $H=\beta H^U_{\Lambda,\mathbf{0}}$, we must calculate the double commutator to find $\langle[A^*,[H,A]]\rangle$. In general spins this is a huge calculation, instead we specialise to the case $S=1$. In this case we can calculate as below, it uses several special properties of the Spin-1 matrices. To make use of this inequality we note that
\begin{equation}
\begin{aligned}
 \mathcal{F}\Bigg\langle\Bigg((S_0^3)^2-\frac{1}{3}&S(S+1)\Bigg)\left((S_x^3)^2-\frac{1}{3}S(S+1)\right)\Bigg\rangle^U_{\Lambda,\mathbf{h}}(k)
 \\
 =&\sum_{x\in\Lambda}e^{-ikx}\left\langle\left((S_0^3)^2-\frac{1}{3}S(S+1)\right)\left((S_x^3)^2-\frac{1}{3}S(S+1)\right)\right\rangle^U_{\Lambda,\mathbf{h}}
 \\
 =&\frac{1}{|\Lambda|}\sum_{x,y\in\Lambda}e^{-ik(x-y)}\left\langle\left((S_x^3)^2-\frac{1}{3}S(S+1)\right)\left((S_y^3)^2-\frac{1}{3}S(S+1)\right)\right\rangle^U_{\Lambda,\mathbf{h}}
 \\
 =&\frac{1}{|\Lambda|}\left\langle\mathcal{F}\left((S_{x}^3)^2-\frac{1}{3}S(S+1)\right)(-k)\mathcal{F}\left((S_y^3)^2-\frac{1}{3}S(S+1)\right)(k)\right\rangle^U_{\Lambda,\mathbf{h}}.
 \end{aligned}
\end{equation}
This relation holds for other correlation functions, including the Duhamel correlation function, but for Duhamel
\begin{equation}
\begin{aligned}
 \mathcal{F}\Bigg((S_0^3)^2-\frac{1}{3}&S(S+1),(S_x^3)^2-\frac{1}{3}S(S+1)\Bigg)_{Duh}(k)
 \\
 =&\frac{1}{|\Lambda|}\left(\mathcal{F}\left((S_{x}^3)^2-\frac{1}{3}S(S+1)\right)(k),\mathcal{F}\left((S_y^3)^2-\frac{1}{3}S(S+1)\right)(k)\right)_{Duh},
 \end{aligned}
\end{equation}
 there is no $-k$ because of the definition of the Duhamel correlation function and the equality $\left(\mathcal{F}\left[(S_x^3)^2\right]\right)(k)^*=\mathcal{F}\left[(S_{x}^3)^2\right](-k)$.

First we prove a preliminary lemma regarding the double commutator
\\
\begin{doublecomm}
For $S=1$, $A=\mathcal{F}\left((S_x^3)^2-\frac{2}{3}\right)(k)$ and $H=\beta H_{\Lambda,0}$ we have 
\[
\langle[A^*,[H,A]]\rangle^U_{\Lambda,\mathbf{h}}=8\beta|\Lambda|\varepsilon(k+\pi)\left\langle S_0^1S_0^3S_{e_1}^1S_{e_1}^3\right\rangle^U_{\Lambda,\mathbf{h}}
\]
where $e_1$ is the first basis vector in $\mathbb{Z}^d$.
\end{doublecomm}
\begin{proof}
The proof is just a calculation, although it is somewhat complicated, we begin by noting that in the case $S=1$ the matrices $(S^i)^2$ and $(S^j)^2$ commute and $(S^i)^3=S^i$ for $i,j=1,2,3$.
\begin{equation}
\begin{aligned}
 \left[H,A\right] =&-2\beta\sum_{x,y:\{x,y\}\in\mathcal{E}}e^{-ikx}\left[(S_x^1S_y^1-S_x^2S_y^2+S_x^3S_y^3)^2,(S_x^3)^2\right]
         \\
         =& -2\beta\sum_{x,y:\{x,y\}\in\mathcal{E}}e^{-ikx}\Big[ (S_x^1S_x^3S_y^1S_y^3-S_x^1S_x^2S_y^1S_y^2-S_x^2S_x^1S_y^2S_y^1
         \\
        &\qquad\qquad\qquad\qquad\qquad -S_x^2S_x^3S_y^2S_y^3+S_x^3S_x^1S_y^3S_y^1-S_x^3S_x^2S_y^3S_y^2),(S_x^3)^2\Big].
\end{aligned}
\end{equation}

The square terms have dropped out as they commute with $(S_x^3)^2$, as does the constant term $S(S+1)/3$. Now we calculate the commutator for each term in the sum, here we make use of the fact that $S^iS^jS^i=0$ for $i\neq j$,  $i,j=1,2,3$ for $S=1$.
\begin{equation}
\begin{aligned}
 \left[H,A\right]=&-2\beta\sum_{x,y:\{x,y\}\in\mathcal{E}}e^{-ikx}\bigg(S_x^1S_x^3S_y^1S_y^3+\overbrace{\left[(S_x^3)^2,S_x^1S_x^2\right]}^{=0}S_y^1S_y^2+\overbrace{\left[(S_x^3)^2,S_x^2S_x^1\right]}^{=0}S_y^2S_y^1
 \\
 &\qquad\qquad\qquad\qquad\qquad\qquad\qquad\,\,-S_x^2S_x^3S_y^2S_y^3-S_x^3S_x^1S_y^3S_y^1+S_x^3S_x^2S_y^3S_y^2\bigg)
 \\
 =&+2\beta\sum_{x,y:\{x,y\}\in\mathcal{E}}e^{-ikx}\bigg(\left[S_x^2S_y^2,S_x^3S_y^3\right]+\left[S_x^3S_y^3,S_x^1S_y^1\right]\bigg).
\end{aligned}
\end{equation}
Now calculating the commutator of these products and using the spin commutation relations we obtain
\begin{equation}
\left[H,A\right]=2\beta i\sum_{x,y:\{x,y\}\in\mathcal{E}}e^{-ikx}\underbrace{\left(S_x^2S_x^3S_y^1+S_x^3S_x^1S_y^2+S_x^1S_y^3S_y^2+S_x^2S_y^1S_y^3\right)}_{f(\mathbf{S}_x,\mathbf{S}_y)}.
\end{equation}
Now we can use this to calculate the double commutator, firstly we split the commutator into the sum of two similar terms
\begin{equation}
\begin{aligned}
\left[A^*,\left[H,A\right]\right]=&2\beta i \sum_{x,y:\{x,y\}\in\mathcal{E}}e^{-ikx}\left[e^{ikx}(S_x^3)^2+e^{iky}(S_y^3)^2,f(\mathbf{S}_x,\mathbf{S}_y)\right]
\\
&=2\beta i \sum_{x,y:\{x,y\}\in\mathcal{E}}\left[(S_x^3)^2, f(\mathbf{S}_x,\mathbf{S}_y)  \right]+\cos(k(x-y))\left[(S_y^3)^2, f(\mathbf{S}_x,\mathbf{S}_y) \right].
\end{aligned}
\end{equation}
We can calculate each of these commutators separately, the first double commutator can be calculated as follows

\begin{equation}
\begin{aligned}
\left[(S_x^3)^2, f(\mathbf{S}_x,\mathbf{S}_y)  \right]=&\left[(S_x^3)^2,S_x^2S_x^3S_y^1+S_x^3S_x^1S_y^2+S_x^1S_y^3S_y^2+S_x^2S_y^1S_y^3\right]
\\
=&-S_x^2S_x^3S_y^1+iS_x^3S_x^2S_y^3S_y^2+iS_x^2S_x^3S_y^3S_y^2
\\
&\qquad\qquad\qquad +S_x^3S_x^1S_y^2-iS_x^3S_x^1S_y^3S_y^1-iS_x^1S_x^3S_y^1S_y^3.
\end{aligned}
\end{equation}
We recognise the commutator relations above to finally give
\begin{equation}
\left[(S_x^3)^2, f(\mathbf{S}_x,\mathbf{S}_y) \right]=iS_x^2S_x^3S_y^2S_y^3+iS_x^3S_x^2S_y^3S_y^2-iS_x^3S_x^1S_y^3S_y^1-iS_x^1S_x^3S_y^1S_y^3.
\end{equation}
For the other commutator we follow the previous calculation almost exactly and in fact we find the two commutators are equal
\begin{equation}
\left[(S_y^3)^2, f(\mathbf{S}_x,\mathbf{S}_y) \right]=\left[(S_x^3)^2, f(\mathbf{S}_x,\mathbf{S}_y) \right].
\end{equation}
To finish the calculation we take expectations
\begin{equation}
\begin{aligned}
\langle[A^*,[H,A]]\rangle^U_{\Lambda,\mathbf{h}}=&
\\
-4\beta |\Lambda|\sum_{i=1}^d(1+&\cos(k_i))\left\langle S_0^2S_0^3S_{e_i}^2S_{e_i}^3+S_0^3S_0^2S_{e_i}^3S_{e_i}^2-S_0^3S_0^1S_{e_i}^3S_{e_i}^1-S_0^1S_0^3S_{e_i}^1S_{e_i}^3\right\rangle^U_{\Lambda,\mathbf{h}}
\end{aligned}
\end{equation}
now use the identities $(S^3S^2)^T=-S^2S^3$ and $(S^3S^1)^T=S^1S^3$ and get
\begin{equation}
\begin{aligned}
\langle[A^*,[H,A]]\rangle^U_{\Lambda,\mathbf{h}}=&-8\beta |\Lambda|\sum_{i=1}^d(1+\cos(k_i))\left\langle S_0^2S_0^3S_{e_i}^2S_{e_i}^3-S_0^3S_0^1S_{e_i}^3S_{e_i}^1\right\rangle^U_{\Lambda,\mathbf{h}}
\\
=&8\beta |\Lambda|\sum_{i=1}^d(1+\cos(k_i))\left\langle 2S_0^2S_0^3S_{e_i}^2S_{e_i}^3\right\rangle^{0,J_2}_{\beta,\Lambda,\mathbf{h}}.
\\
=&8\beta|\Lambda|\varepsilon(k+\pi)\left\langle S_0^1S_0^3S_{e_1}^1S_{e_1}^3\right\rangle^{0,J_2}_{\beta,\Lambda,\mathbf{h}}.
\end{aligned}
\end{equation}
On the second line we have used that $US_{e_1}^2S_{e_1}^3=-S_{e_1}^2S_{e_1}^3U$ to move from states $\langle\cdot\rangle^U_{\Lambda,\mathbf{h}}$ to states $\langle\cdot\rangle^{0,J_2}_{\beta,\Lambda,\mathbf{0}}$ and on the third line we have used that each cross term $\langle S_x^iS_x^jS_y^iS_y^j\rangle^{0,J_2}_{\beta,\Lambda,\mathbf{h}}$ has the same expectation value. Now simply note that the above correlation is the same in $\langle \cdot\rangle^U_{\Lambda,\mathbf{h}}$ and in $\langle\cdot\rangle^{0,J_2}_{\beta,\Lambda,\mathbf{0}}$.
\end{proof}
 Using this in Falk-Bruch we have the bound
 \begin{equation}\label{IRB2}
 \hat{\rho}(k)\leq \sqrt{\left\langle S_0^1S_0^3S_{e_1}^1S_{e_1}^3\right\rangle^{0,J_2}_{\beta,\Lambda,\mathbf{h}}}\sqrt{\frac{\varepsilon(k+\pi)}{\varepsilon(k)}}+\frac{1}{2\beta\varepsilon(k)}.
 \end{equation}
 The possibility of obtaining a result is not ruled out for other values of $S$, I expect it to be the case for other values of $S$, but computing the double commutator in Falk-Bruch becomes extremely complicated.

Now using the Fourier transform in the following way:
 \begin{equation}
 \left\langle\left((S_0^3)^2-\frac{2}{3}\right)\left((S_y^3)^2-\frac{2}{3}\right)\right\rangle^{0,J_2}_{\Lambda,\mathbf{h}}=\frac{1}{|\Lambda|}\sum_{x\in\Lambda}\rho(x)+\frac{1}{|\Lambda|}\sum_{k\in\Lambda^*\setminus\{0\}}e^{ik\cdot y}\hat{\rho}(x)(k)
 \end{equation}
 
with $y=e_1$ we get the lower bound
\begin{equation}
\frac{1}{|\Lambda|}\sum_{x\in\Lambda}\rho(x)\geq\rho(e_1)-\frac{1}{|\Lambda|}\sum_{k\in\Lambda^*\setminus\{0\}}\sqrt{\frac{\varepsilon(k+\pi)}{\varepsilon(k)}}\left(\frac{1}{d}\sum_{i=1}^d\cos k_i\right)_+-\frac{1}{2\beta\varepsilon(k)}.
\end{equation}
Taking the thermodynamic limit with $L_i=L$ even for $i=1,..,d$ gives
\begin{equation}\label{disclb}
\liminf_{L\to\infty}\left\langle\frac{1}{|\Lambda|}\sum_{x\in\Lambda}\rho(x)\right\rangle\geq\rho(e_1)-\frac{1}{(2\pi)^d}\int_{[-\pi,\pi]^d}\left(\sqrt{\frac{\varepsilon(k+\pi)}{\varepsilon(k)}}\left(\frac{1}{d}\sum_{i=1}^d\cos k_i\right)_++\frac{1}{2\beta\varepsilon(k)}\right)\mathrm{d}k.
\end{equation}
 The integral is finite if and only if $d\geq3$ due to the last term. Now taking the limit $\beta\to\infty$ gives the result.
 \qed
 
\section{Extending to $J_1<0$}
 
The aim of this section is to extend the proof of theorem 3.1 to a proof of theorem 2.1. The proof of long-range order for $J_1<0$ is a straightforward extension of the previous results, like before we will work with a Hamiltonian that is Unitarily equivalent to $H^{J_1,J_2}_{\Lambda,\mathbf{0}}$, we also introduce an external field $\mathbf{h}$ as before. Recall the unitary operator $U=\prod_{x\in\Lambda_B}e^{i\pi S_x^2}$, let
\begin{equation}
\widetilde{H}^U_{\Lambda,\mathbf{h}}=UH^{J_1,J_2}_{\Lambda,\mathbf{0}}U^{-1}-\sum_{x\in\Lambda}h_x\left((S_x^3)^2-\frac{1}{3}S(S+1)\right).
\end{equation}
 The effect of the unitary operator here is to replace $S_x^1$ and $S_x^3$ in $H^{J_1,J_2}_{\Lambda,\mathbf{0}}$ with $-S_x^1$ and $-S_x^3$ respectively. By using the representation \eqref{Qspin-matrix} we can write $\widetilde{H}^U_{\Lambda,\mathbf{0}}$ as
 \begin{equation}
 \begin{aligned}
 \widetilde{H}^U_{\Lambda,\mathbf{0}}=-\sum_{\{x,y\}\in\mathcal{E}}\bigg[J_1\Big((S_x^1-S_y^1)^2-&(S_x^2-S_y^2)^2+(S_x^3-S_y^3)^2\Big)
 \\
 &-J_2\Big(\mathcal{TR}\big[(Q_x-Q_y)^2\big]\Big)+C_\Lambda(J_1,J_2)\bigg].
 \end{aligned}
 \end{equation}
 Then similar to before we introduce the field $v$ and associated $3\times 3$ field of matrices $\mathbf{v}$, define
 \begin{align}
 \widetilde{H}(v)=&-\sum_{\{x,y\}\in\mathcal{E}}\bigg[J_1\Big((S_x^1-S_y^1)^2-(S_x^2-S_y^2)^2+(S_x^3-S_y^3)^2\Big)
 \\
 &\qquad\quad-J_2\Big(\mathcal{TR}\big[(Q_x+\frac{\mathbf{v}_x}{2}-Q_y-\frac{\mathbf{v}_y}{2})^2\big]\Big)+C_\Lambda(J_1,J_2)\bigg]-\frac{1}{4}(v,-\Delta v),\nonumber
 \\
 \widetilde{Z}(v)=&Tre^{-\beta \widetilde{H}(v)},
 \end{align}
 and
 \begin{align}
 \widetilde{H}'(v)=&\widetilde{H}(v)+\frac{1}{4}(v,-\Delta v),
 \\
 \widetilde{Z}'(v)=&Tr e^{-\beta \widetilde{H}'(v)}.
 \end{align}
 
 From this reflection positivity follow just as in Lemma 3.4, with the obvious changes to $A$ and $B$ and the extra terms
 \begin{equation}
 \begin{array}{l l}
 C^{10}_i=\sqrt{-J_1}S^1_{x_i}, \quad\quad & D^{10}_i=\sqrt{-J_1}S^1_{y_i},
 \\
 C^{11}_i=\sqrt{-J_1}iS^2_{x_i},\quad\quad & D^{11}_i=\sqrt{-J_1}iS^2_{y_i},
 \\
 C^{12}_i=\sqrt{-J_1}S^3_{x_i},\quad\quad & D^{12}_i=\sqrt{-J_1}S^3_{y_i},
 \end{array}
 \end{equation}
 (recall that $J_1<0$). From this we obtain the Gaussian domination inequality
 \begin{equation}
 \widetilde{Z}(v)\leq \widetilde{Z}(0)e^{\frac{\beta}{4}(v,-\Delta v)} \iff \widetilde{Z}'(v)\leq \widetilde{Z}'(0),
 \end{equation}
 just as before. We also obtain the same infrared bound as in Lemma 3.7, with an identical proof
 \begin{equation}
 \mathcal{F}\left((S_0^3)^2-\frac{1}{3}S(S+1),(S_x^3)^2-\frac{1}{3}S(S+1)\right)_{Duh}\leq \frac{1}{2\beta\varepsilon(k)}.
 \end{equation}
 Again the results up to here work for general $S\in\frac{1}{2}\mathbb{N}$, at this point we must specialise to $S=1$ to be able to calculate the quantities in the double commutator of the Falk-Bruch inequality.
 From this we can see that by using Falk-Bruch inequality with $A=\mathcal{F}\left((S_x^3)^2-\frac{2}{3}\right)(k)$ and $H=\beta \widetilde{H}^U_{\Lambda,\mathbf{0}}$ the linearity of the double commutator means that there will be an extra term in the analogous result to lemma 3.6 equal to $\langle J_1[A^*,[-2\sum_{\{x,y\}\in\mathcal{E}}(\mathbf{S}_x\cdot\mathbf{S}_y),A]]\rangle$. This will result in the IRB analogous to \eqref{IRB2} potentially being larger, weakening the result. If $|J_1|$ is small enough this weakening will not be too severe so as to make the lower bound analogous to the bound in theorem 3.1 negative in cases where we know the original lower bound was positive. This ensures that we have a positive lower bound $C=C(\beta,J_1)$ in Theorem 2.1 when $\beta$ and $|J_1|$ are small enough. It is worth noting that for the same reason as just described, extending the result of Dyson, Lieb and Simon \cite{D-L-S} to $J_2>0$ also requires that $|J_2|$ is small. This means the two results will not 
overlap, leaving part of the quadrant $J_1\leq0\leq J_2$ still open to investigation.
 
 \subsection*{Acknowledgments}
 
 I am pleased to thank my supervisor Daniel Ueltschi for his support and useful discussions. I am also grateful to the referee for several useful comments and observations. This work is supported by EPSRC as part of the MASDOC DTC at the University of Warwick. Grant No. EP/HO23364/1.

\end{document}